\documentclass[11pt]{article}
\pdfoutput=1
\usepackage{ifthen} 
\newboolean{twocolumn} \makeatletter \if@twocolumn \setboolean{twocolumn}{true} \else \setboolean{twocolumn}{false} \fi \makeatother  
\newcommand{\ifelsetwocolumn}[2]{{\ifthenelse {\boolean{twocolumn}} {#1} {#2} }}

\iftwocolumn \usepackage{latex8, times} \fi 

\usepackage{fullpage}
\usepackage{amsfonts, amssymb, amsmath, amsthm}
\usepackage{latexsym}
\usepackage{url}
\usepackage{xcolor}
\definecolor{DarkGray}{rgb}{0.1,0.1,0.5}
\usepackage[colorlinks=true,breaklinks, linkcolor= DarkGray,citecolor= DarkGray,urlcolor= DarkGray]{hyperref}	
\usepackage{graphicx}
\usepackage[tight]{subfigure}
\usepackage{cancel}
\usepackage{multirow}
\usepackage{caption}	

\usepackage{import}
\usepackage{ltxtable}

\usepackage{calc}	

\usepackage[tracking=smallcaps]{microtype}	

\newcommand{\bra}[1]{{\langle#1|}}
\newcommand{\ket}[1]{{|#1\rangle}}
\newcommand{\braket}[2]{{\langle#1|#2\rangle}}
\newcommand{\ketbra}[2]{{\ket{#1}\!\bra{#2}}}
\newcommand{\abs}[1]{{\lvert #1\rvert}}	
\newcommand{\norm}[1]{{\lVert #1 \rVert}}
\newcommand{\bignorm}[1]{{\big\| #1 \big\|}}
\newcommand{\Bignorm}[1]{{\Big\| #1 \Big\|}}

\newcommand{\Tr}{\mathrm{Tr}}

\def\C {{\bf C}}
\def\Re {{\mathrm{Re}}}
\def\D {{\mathcal D}}
\def\H {{\mathcal H}}
\def\L {{\mathcal L}}
\def\N {{\bf N}}

\newcommand{\identity}{\ensuremath{\boldsymbol{1}}} 

\newcommand{\ADVpm} {\mathrm{Adv}^{\pm}}

\newcommand{\B}{\{0,1\}}	

\def\u{\mu}
\def\ut{\nu}

\def\Pim {{\overline{\Pi}}}

\DeclareMathOperator{\gtwoop}{{\operatorname{\gamma_2}}}
\newcommand{\gtwo}[2]{{\gtwoop({#1}\vert{#2})}}
\newcommand{\gtwoDelta}[1]{\gtwo{#1}{\Delta}}
\newcommand{\gtwostar}[2]{{\gamma_2^*({#1}\vert{#2})}}

\newcounter{sprows}

\newlength{\spheight}
\newlength{\spraise}

\newcommand{\comment}[1]{\emph{\color{blue}Comment:\color{black} #1}}
\newlength{\commentslength}
\newcommand{\comments}[1]{
\hspace{-2\parindent}
\addtolength{\commentslength}{-\commentslength}
\addtolength{\commentslength}{\linewidth}
\addtolength{\commentslength}{-\parindent}
\fcolorbox{blue}{white}{\smallskip\begin{minipage}[c]{\commentslength}
\emph{Comments:}\begin{itemize}#1\end{itemize}\end{minipage}}\bigskip
}
\renewcommand{\comment}[1]{}\renewcommand{\comments}[1]{}    
\newcommand{\rem}[1]{}
\newcommand{\ignore}[1]{}

\newtheorem{theorem}{Theorem}[section]
\newtheorem{lemma}[theorem]{Lemma}
\newtheorem{corollary}[theorem]{Corollary}
\newtheorem{claim}[theorem]{Claim}
\newtheorem{proposition}[theorem]{Proposition}

\newtheorem{fact}[theorem]{Fact}

\newtheorem{definition}[theorem]{Definition}

\newfont{\subsubsecfnt}{ptmri8t at 10pt}
\renewcommand{\subparagraph}[1]{\smallskip{\subsubsecfnt #1.}}

\numberwithin{equation}{section} 

\newcommand{\eqnref}[1]{\hyperref[#1]{{(\ref*{#1})}}}
\newcommand{\thmref}[1]{\hyperref[#1]{{Theorem~\ref*{#1}}}}
\newcommand{\lemref}[1]{\hyperref[#1]{{Lemma~\ref*{#1}}}}
\newcommand{\corref}[1]{\hyperref[#1]{{Corollary~\ref*{#1}}}}
\newcommand{\defref}[1]{\hyperref[#1]{{Definition~\ref*{#1}}}}
\newcommand{\secref}[1]{\hyperref[#1]{{Section~\ref*{#1}}}}
\newcommand{\figref}[1]{\hyperref[#1]{{Figure~\ref*{#1}}}}
\newcommand{\tabref}[1]{\hyperref[#1]{{Table~\ref*{#1}}}}
\newcommand{\factref}[1]{\hyperref[#1]{{Fact~\ref*{#1}}}}
\newcommand{\remref}[1]{\hyperref[#1]{{Remark~\ref*{#1}}}}
\newcommand{\appref}[1]{\hyperref[#1]{{Appendix~\ref*{#1}}}}
\newcommand{\claimref}[1]{\hyperref[#1]{{Claim~\ref*{#1}}}}
\newcommand{\propref}[1]{\hyperref[#1]{{Proposition~\ref*{#1}}}}
\newcommand{\exampleref}[1]{\hyperref[#1]{{Example~\ref*{#1}}}}
\newcommand{\conjref}[1]{\hyperref[#1]{{Conjecture~\ref*{#1}}}}

\allowdisplaybreaks[1]

\begin{document}

\title{Quantum query complexity of state conversion}
\author{Troy Lee \and Rajat Mittal \and Ben W.~Reichardt \and Robert {\v S}palek \and Mario Szegedy}

\date{}

\maketitle

\begin{abstract}
State conversion generalizes query complexity to the problem of converting between two input-dependent quantum states by making queries to the input.  We characterize the complexity of this problem by introducing a natural information-theoretic norm that extends the Schur product operator norm.  The complexity of converting between two systems of states is given by the distance between them, as measured by this norm.  

In the special case of function evaluation, the norm is closely related to the general adversary bound, a semi-definite program that lower-bounds the number of input queries needed by a quantum algorithm to evaluate a function.  We thus obtain that the general adversary bound characterizes the quantum query complexity of any function whatsoever.  This generalizes and simplifies the proof of  the same result in the case of boolean input and output.  Also in the case of function evaluation, we show that our norm satisfies a remarkable composition property, implying that the quantum query complexity of the composition of two functions is at most the product of the query complexities of the functions, up to a constant.  Finally, our result implies that discrete and continuous-time query models are equivalent in the bounded-error setting, even for the general state-conversion problem.  
\end{abstract}

\section{Introduction}

A quantum query algorithm for evaluating a function~$f$ attempts to compute $f(x)$ with as few queries to the input~$x$ as possible.  Equivalently, the algorithm begins in a state $\ket 0$, and should approach $\ket{f(x)} \otimes \ket 0$.  The \emph{state-conversion} problem generalizes function evaluation to the case where the algorithm begins in a state $\ket{\rho_x}$ and the goal is to convert this to~$\ket{\sigma_x}$.  State-conversion problems arise naturally in algorithm design, generalizing classical subroutines (\figref{f:potatoes}).  For example, the graph isomorphism problem can be reduced to creating a certain quantum state~\cite{Shi02collision}.  

We characterize the quantum query complexity of state conversion.  We introduce a natural, information-theoretic norm, which extends the Schur product operator norm.  The complexity of state conversion depends only on the Gram matrices of the sets of vectors~$\{\ket{\rho_x}\}$ and~$\{\ket{\sigma_x}\}$, and is characterized as the norm of the {difference} between these Gram matrices.  For example, in function evaluation, the initial Gram matrix is the all-ones matrix, $J$, and the target Gram matrix is $F = \{ \delta_{f(x), f(y)} \}_{x,y}$, so the query complexity depends only on the norm of $F-J$.  Characterizing query complexity in terms of a norm-induced metric has interesting consequences.  For example, it follows that if one can design an optimal algorithm for going from $J$ to $\tfrac{99}{100} J + \tfrac{1}{100} F$, then one also obtains an optimal algorithm for evaluating~$f$.  
 
The norm we introduce is related to the general adversary bound~\cite{HoyerLeeSpalek07negativeadv}, a strengthening of the popular adversary method for showing lower bounds on quantum query complexity~\cite{Ambainis00adversary}.  A recent sequence of works~\cite{fgg:and-or,ccjy:and-or, AmbainisChildsReichardtSpalekZhang07andor, ReichardtSpalek08spanprogram} has culminated in showing that the general adversary bound gives, up to a constant factor, the bounded-error quantum query complexity of any function with boolean output and binary input alphabet~\cite{Reichardt09spanprogram_arxivandfocs, Reichardt10advtight}.  Our more general state-conversion result completes this picture by showing that the general adversary bound characterizes the bounded-error quantum query complexity of any function whatsoever: 

\begin{theorem} \label{t:advtight} 
Let $f : \D \rightarrow E$, where $\D \subseteq D^n$, and $D$ and $E$ are finite sets.  Then the bounded-error quantum query complexity of $f$, $Q(f)$, is characterized by the general adversary bound, $\ADVpm(f)$: 
\begin{equation}
Q(f) = \Theta\big(\ADVpm(f)\big)
 \enspace .
\end{equation}
\end{theorem}

The general adversary bound is a semi-definite program (SDP).  When phrased as a minimization problem, $\ADVpm(f)$ only has constraints on $x, y$ pairs where $f(x) \ne f(y)$.  In contrast, we consider an SDP that places constraints on {\em all} input pairs~$x, y$.  Fortunately, these extra constraints increase the optimal value by at most a factor of two.  The extra constraints, however, are crucial for the construction of the algorithm, and for any extension to state conversion.  They also lead to a new conceptual understanding of the adversary bound.  

The modified SDP defines a norm, and we define the \emph{query distance} as the metric induced by this new norm.  Our main algorithmic theorem states that there is a quantum algorithm that converts~$\ket{\rho_x}$ to a state with high fidelity to $\ket{\sigma_x}$, and that makes a number of queries of order the query distance between $\rho$ and $\sigma$, the respective Gram matrices of $\{\ket{\rho_x}\}$ and $\{\ket{\sigma_x}\}$.  

The correctness of our algorithm has a direct and particularly simple proof.  Though more general, it simplifies the previous characterization of boolean function evaluation.  At its mathematical heart is a lemma that gives an ``effective" spectral gap for the product of two reflections.  

The query distance also gives lower bounds on the query complexity of state conversion.  It is straightforward to argue that the query distance between $\rho$ and $\sigma$ lower bounds the number of queries needed to reach $\sigma$ {\em exactly}.  To deal with the bounded-error case, we look at the minimum over all $\sigma'$ of the query distance between $\rho$ and $\sigma'$, where $\sigma'$ is a valid Gram matrix for the final states of a successful algorithm.  We show that a simpler necessary and sufficient condition for the latter is that $\sigma$ and $\sigma'$ are close in the distance induced by the Schur product operator norm.   

As the query distance remains a lower bound on the continuous-time query complexity, a corollary of our algorithm is that the continuous-time and discrete query models are related up to a constant factor in the bounded-error setting.  Previously, this equivalence was known up to a sub-logarithmic factor~\cite{CleveGottesmanMoscaSommaYongeMallo08discretize}.  

\begin{figure}
\centering
\includegraphics{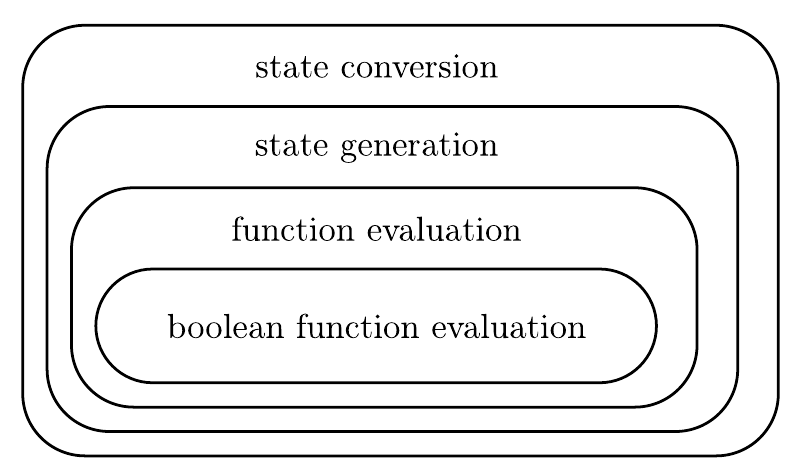}
\caption{The state-conversion problem generalizes the state-generation problem, studied by Ambainis et al.~\cite{AmbainisMagninRoettelerRoland10stategeneration}, which in turn generalizes function evaluation.  The quantum query complexity of evaluating boolean functions has been characterized by~\cite{Reichardt10advtight}.} \label{f:potatoes}
\end{figure}

Since the general adversary bound characterizes quantum query complexity, its properties immediately carry over thereto.  We show that the general adversary bound satisfies a remarkable composition property, that $\ADVpm$ of a composed function $f \circ (g, g, \ldots, g)$ is $O\big(\ADVpm(f) \ADVpm(g)\big)$.  Previously this was known in the boolean case~\cite{Reichardt09spanprogram_arxivandfocs}, and, again, having constraints on all input pairs turns out to be crucial in the extension to non-boolean functions.  When the input of~$f$ is boolean, we can show a matching lower bound, extending~\cite{HoyerLeeSpalek07negativeadv}.

\section{Background} \label{s:definitions}

For a natural number $n$, let $[n] = \{1, 2, \ldots, n\}$.  For two matrices $A$, $B$ of the same size, $A \circ B$ denotes their entrywise product, also known as Schur or Hadamard product.  Let $\langle A, B \rangle = \Tr (A^\dagger B)$.  Denote by $\norm A$ the spectral norm of~$A$.  We will use $\identity$ and $J$ for the identity and all-ones matrices, respectively, where size can be inferred from the context.  Let $\delta_{a,b}$ be the Kronecker delta function.

\subsection{Coherent and non-coherent state-conversion problems}

The \emph{quantum query complexity} of a function $f$, $Q(f)$, is the number of input queries needed to evaluate $f$ with error at most $1/3$~\cite{BuhrmanDeWolf02querysurvey}.  The \emph{state-conversion} problem generalizes function evaluation to the case where the aim is  to transform one quantum state into another, using queries to the input.  The problem is parameterized by sets of states $\{ \ket{\rho_x} \}$ and $\{ \ket{\sigma_x} \}$.  On input $x$, we begin in a state~$\ket{\rho_x}$ and wish to create the state $\ket{\sigma_x}$, using as few queries to $x$ as possible.  State conversion is a slight generalization of the \emph{state-generation} problem, in which in the initial state $\ket{\rho_x}$ is independent of~$x$.  This problem was introduced by Shi~\cite{Shi02collision}, and recently studied systematically by Ambainis et al.~\cite{AmbainisMagninRoettelerRoland10stategeneration}.  

State conversion has two variants, coherent and non-coherent.  Both versions allow the algorithm workspace.  

\begin{definition}[Coherent output condition] \label{def_coherent}
An algorithm solves the coherent state-conversion problem with error $\epsilon$ if for every $x$ it replaces the initial state $\ket{\rho_x} \otimes \ket 0$ by a state $\ket{\sigma'_x}$ such that $\Re(\bra{\sigma'_x}(\ket{\sigma_x}\otimes\ket 0)) \ge \sqrt{1-\epsilon}$.  
\end{definition}

\begin{definition}[Non-coherent output condition] \label{def_nc}
An algorithm solves the non-coherent state-conversion problem with error $\epsilon$ if for every $x$ it replaces the initial state $\ket{\rho_x} \otimes \ket 0$ by a state $\ket{\sigma'_x}$ such that $|\bra{\sigma'_x}(\ket{\sigma_x}\otimes\ket{s_x})| \ge \sqrt{1-\epsilon}$ for some state $\ket{s_x}$ that may depend on $x$.  
\end{definition}

The query complexity of state conversion only depends on the Gram matrices of the initial and target states, i.e., on $\rho = \{ \braket{\rho_x}{\rho_y} \}_{x,y}$ and $\sigma = \{ \braket{\sigma_x}{\sigma_y} \}_{x,y}$.  (In function evaluation and state generation, $\rho$ is the all-ones matrix.)  Let $Q_\epsilon(\rho, \sigma)$ and $Q^{nc}_\epsilon(\rho, \sigma)$ be the minimum number of queries required to solve the coherent and non-coherent state-conversion problems, respectively, with error~$\epsilon$.  

For evaluating functions, coherent and non-coherent complexities are equal up to constant factors.  An example that illustrates the difference between these output conditions is computing a boolean function in the phase, that is where the target state $\ket{\sigma_x} = (-1)^{f(x)} \ket 0$.  In this case the non-coherent complexity is trivial, while the bounded-error coherent complexity is equal to~$Q(f)$, up to constant factors.

\subsection{The \texorpdfstring{$\gamma_2$}{gamma\_2} norm}

We will make use of the $\gamma_2$ norm, also known as the Schur product operator norm~\cite{Bhatia07, LeeShraibmanSpalek08}.  This norm has been introduced recently to complexity theory by Linial et al.~\cite{LinialMendelsonSchechtmanShraibman07}, and has proven to be very useful in quantum information.  It is currently the best lower bound known on quantum communication~\cite{LinialShraibman09cc}, and Tsirelson has shown that its dual norm characterizes the bias of a quantum XOR game~\cite{Tsirelson87xor, Unger08thesis}.  

\begin{definition}
Let $A$ be matrix with rows labeled by $\D_1$ and columns by $\D_2$.  Define 
\begin{equation} \label{e:gamma2def}
\gamma_2(A) = \!\!\! \min_{\Large \substack{m \in \N, \\ \ket{u_{x}}, \ket{v_{y}} \in \C^m}} \!\!\!
\Big\{
\max \big\{\max_{x \in \D_1} \norm{\ket{u_x}}^2, \max_{y \in \D_2} \norm{\ket{v_y}}^2 \big\} : 
\text{$\forall x \in \D_1, y \in \D_2$, $A_{x,y} = \braket{u_x}{v_y}$}
\Big\} \enspace .
\end{equation}
\end{definition}

The following fact plays a key role in the design of our algorithm and in relating our new norm to the general adversary bound: 

\begin{fact} \label{t:gamma2}
For any $k \in \N$, let $\identity$ and $J$ be the $k$-by-$k$ identity and all-ones matrices, respectively.  Then $\gamma_2(J-\identity) \leq 2(1 - 1/k)$.  
\end{fact}

\begin{proof}
We demonstrate unit vectors $\{\ket{\u_i}\}_{i \in [k]}$, and $\{\ket{\ut_i}\}_{i \in [k]}$ such that 
$\braket{\u_i}{\ut_j} = \tfrac12 \tfrac{k}{k-1}(1-\delta_{i,j})$: let $\ket{\u_i} = -\alpha \ket i + {\sqrt{1-\alpha^2} \over \sqrt{k-1}} \sum_{j \neq i} \ket j$, $\ket{\ut_i} = \sqrt{1-\alpha^2} \ket i + {\alpha \over \sqrt{k-1}} \sum_{j \neq i} \ket j$, for $\alpha = \sqrt{\frac12 - \frac{\sqrt{k-1}}{k}}$.  
\end{proof}

\section{Filtered \texorpdfstring{$\gamma_2$}{gamma\_2} norm and query distance}

We define a natural generalization of the $\gamma_2$ norm, in which the factorization is filtered through certain matrices: 

\begin{definition}[Filtered $\gamma_2$ norm]
Let $A$ be a matrix with rows indexed by elements of $\D_1$ and columns by $\D_2$, and let $Z = \{Z_1, \ldots, Z_n\}$ be a set of $|\D_1|$-by-$|\D_2|$ matrices.  Define $\gtwo{A}{Z}$ by 
\begin{equation} \label{e:filteredgamma2def}
\begin{split}
\gtwo{A}{Z} &= \min_{\Large \substack{m \in \N, \\ \ket{u_{x j}}, \ket{v_{y j}} \in \C^m}}  \max \Big\{
\max_{x \in \D_1} \sum_j \norm{\ket{u_{xj}}}^2, 
\max_{y \in \D_2} \sum_j \norm{\ket{v_{yj}}}^2\Big\} \\
&\qquad 
\text{$\forall x \in \D_1, y \in \D_2$, $A_{x,y} = \sum_j (Z_j)_{x,y} \braket{u_{xj}}{v_{yj}}$}
 \enspace .
\end{split}
\end{equation}
\end{definition}

The filtered $\gamma_2$ norm $\gtwo{\,\cdot\,}{Z}$ is a norm.  Among its many properties (see \appref{s:gamma2properties}) are that $\gamma_2(A) = \gtwo{A}{\{J\}}$, where $J$ is the all-ones matrix, and $\gtwo{A}{\{A\}} = 1$ if $A \neq 0$.  We use below the general inequality 
\begin{equation}\label{e:gtwocirc} 
\gtwo{A}{\{Z_j\}} \leq \gtwo{A}{\{Z_j \circ B\}} \gamma_2(B)
 \enspace .
\end{equation}  

The query distance is the metric induced when the filter matrices are related to the query~process.  Let $\D \subseteq D^n$ be a finite set, and let $\Delta = \{\Delta_1, \ldots, \Delta_n\}$, where $\Delta_j = \{1 - \delta_{x_j, y_j}\}_{x,y \in \D}$.  Thus $\Delta_j$ encodes when a query to index~$j$ distinguishes input $x$ from input~$y$.  

\begin{definition}
The \emph{query distance} between $\rho$ and $\sigma$, two $\abs\D$-by-$\abs\D$ matrices, is $\gtwoDelta{\rho - \sigma}$.  
\end{definition}

\thmref{t:nonnonnon} below shows that the query distance characterizes the quantum query complexity of state conversion.  Furthermore, as we show now, it is closely related to the general adversary bound, a lower bound on the quantum query complexity for function evaluation introduced by~\cite{HoyerLeeSpalek07negativeadv}.  Let $f: \D \rightarrow E$ and let $F = \{\delta_{f(x), f(y)}\}_{x,y}$.  

\begin{definition}
The general adversary bound for $f$ is given by 
\begin{gather}
\ADVpm(f) = 
\max \Big\{ \norm{\Gamma} : \forall j \in [n], \norm{\Gamma \circ \Delta_j} \leq 1 \Big\} \label{e:ADVpmmax} 
 \enspace ,
\end{gather}
where the maximization is over $\abs\D$-by-$\abs\D$ real, symmetric matrices $\Gamma$ satisfying $\Gamma \circ F = 0$.  
\end{definition}

By taking the dual of this SDP, we obtain a bound that is the same as $\gtwoDelta{J-F}$, except without any constraints on pairs $x, y$ with $f(x) = f(y)$.  In other words, $\ADVpm(f) = \gtwo{J-F}{\{\Delta_j \circ (J-F)\}}$.  

\begin{theorem} \label{cor:gamma2_adversary}
The values of the general adversary bound and $\gtwoDelta{J-F}$ differ by at most a factor of two, and are equal in the case that the function has boolean output: 
\begin{equation*}
\ADVpm(f) \le \gtwoDelta{J-F} \le 2 \big(1-{1}/{\abs E}\big) \ADVpm(f)
 \enspace .
\end{equation*}
\end{theorem}

\begin{proof}
Since $\ADVpm(f)$ has fewer constraints as a minimization problem, $\ADVpm(f) \leq \gtwoDelta{J - F}$.  

For the other direction, use Eq.~\eqnref{e:gtwocirc} with $Z_j = \Delta_j$, $A = B = J-F = A \circ B$.  As it is readily seen that $\gamma_2$ is invariant under 
adding or removing duplicate rows or columns, \factref{t:gamma2} implies $\gamma_2(J-F) \leq 2(1 - 1/\abs E)$.  
\end{proof}

\section{Characterization of quantum query complexity}

In this section, we show that the query complexities for function evaluation, and coherent and non-coherent state conversion are characterized in terms of~$\gamma_2$ and~$\gtwoDelta{\cdot}$.  We begin with the upper bounds.

\subsection{Quantum query algorithm for state conversion} \label{s:nonbooleanextension}

\begin{theorem} \label{t:nonnon}
Consider the problem of converting states $\{\ket{\rho_x}\}$ to $\{\ket{\sigma_x}\}$, for $x \in \D \subseteq D^n$.  Let $\rho$ and $\sigma$ be the states' Gram matrices.  For any $\epsilon \in (0, \gtwoDelta{\rho - \sigma})$, this problem has query complexity 
\begin{equation*}
Q_\epsilon(\rho, \sigma) = O\Big( \gtwoDelta{\rho - \sigma} \frac{\log(1/\epsilon)}{\epsilon^2} \Big)
 \enspace .
\end{equation*}
\end{theorem}

\thmref{t:advtight} follows from Theorems~\ref{cor:gamma2_adversary} and~\ref{t:nonnon}, together with the lower bound from~\cite{HoyerLeeSpalek07negativeadv}.  

\smallskip

The mathematical heart of our analysis is to study the spectrum of the product of two reflections.  The following lemma gives an ``effective" spectral gap for two reflections applied to a vector.  It is closely related to~\cite[Theorem~8.7]{Reichardt09spanprogram_arxivandfocs}, but has a significantly simpler statement and proof.  

\begin{lemma}[Effective spectral gap lemma] \label{t:ref_lemma}
Let $\Pi$ and $\Lambda$ be projections, and let $R = (2\Pi - \identity) (2\Lambda - \identity)$ be the product of the reflections about their ranges.  Let $\{\ket \beta\}$ be a complete orthonormal set of eigenvectors of $R$, with respective eigenvalues $e^{i \theta(\beta)}$, $\theta(\beta) \in (-\pi, \pi]$.  

For any $\Theta \geq 0$, let $P_\Theta = \sum_{\beta : \abs{\theta(\beta)} \le \Theta} \ketbra{\beta}{\beta}$.  If $\Lambda \ket w = 0$, then 
\[
\norm{P_\Theta \Pi \ket w} \le \tfrac{\Theta}{2} \norm{\ket w}
 \enspace .
\]
\end{lemma}

\begin{proof}
The claim can be shown via Jordan's Lemma~\cite{Jordan75projections}; we give a direct proof.  
\def\Pim{{\overline\Pi}}
Let $\ket v = P_\Theta \Pi \ket w$, $\ket{v'} = (2\Lambda - \identity) \ket v$ and $\ket{v''} = (2\Pi - \identity) \ket{v'} = R \ket v$.  When $\Theta$ is small, $\ket v$ and $\ket{v''}$ are close: 
\begin{equation*}
\norm{\ket v - \ket{v''}}{}^2 = \Bignorm{\sum_{\beta : \abs{\theta(\beta)} \le \Theta} (1-e^{i \theta(\beta)}) \braket \beta v \ket \beta}^2 
\le 2(1 - \cos \Theta) \norm{\ket v}^2 \le \Theta^2 \norm{\ket v}^2 
 \enspace .
\end{equation*}
Notice that $\ket v + \ket{v'}$ is fixed by~$\Lambda$.  Similarly, $\Pi$ fixes $\ket{v'} + \ket{v''}$ and~$\Pim = \identity - \Pi$ fixes $\ket{v'} - \ket{v''}$.  Hence $0 = \braket{v+v'}{w} = \bra{v+v'}\Pi \ket{w} + \bra{v+v'}\Pim \ket{w} = \bra{v+v''}\Pi \ket{w} + \bra{v-v''}\Pim \ket{w}$.  
Therefore, $\norm{\ket v}^2 = \abs{\bra v \Pi \ket w} = \tfrac12 \abs{\bra{v - v''} \Pi \ket w + \bra{v + v''} \Pi \ket w} = \frac{1}{2} \abs{\bra{v - v''}(\Pi - \Pim) \ket w}$.   We conclude 
\[
\norm{\ket v}^2 
\le \tfrac12 \norm{\ket v - \ket{v''}} \norm{(\Pi-\Pim) \ket w} 
\le \tfrac{\Theta}{2} \norm{\ket v} \norm{\ket w}
 \enspace . \qedhere
\]
\end{proof}

We will also use a routine that, roughly, reflects about the eigenvalue-one eigenspace of a~unitary: 

\begin{theorem}[Phase detection~\cite{Kitaev95stabilizer, MagniezNayakRolandSantha07search}] \label{t:phasedetection}
For any $\Theta, \delta > 0$, there exists~$b = O(\log\frac1\delta \log \frac1\Theta)$ and, for any unitary $U \in \L(\H)$, a quantum circuit $R(U)$ on $\H \otimes (\C^2)^{\otimes b}$ that makes at most $O\big(\frac{\log(1/\delta)}{\Theta}\big)$ controlled calls to $U$ and $U^{-1}$, and such that for any eigenvector $\ket \beta$ of $U$, with eigenvalue $e^{i \theta}$, $\theta \in (-\pi, \pi]$, 
\begin{itemize}
\item
If $\theta = 0$, then $R(U) \ket \beta \otimes \ket{0^b} = \ket \beta \otimes \ket{0^b}$.  
\item
If $\abs \theta > \Theta$, then $R(U) \ket \beta \otimes \ket{0^b} = - \ket \beta \otimes (\ket{0^b} + \ket{\delta_\beta})$ for some vector $\ket{\delta_\beta}$ with $\norm{\ket{\delta_\beta}} < \delta$.  Thus, if $\ket \gamma \in \H$ is orthogonal to all eigenvectors of $U$ with eigenvalues $e^{i \theta}$ for $\abs \theta \leq \Theta$, then $\bignorm{ (R(U) + \identity) \ket \gamma \otimes \ket{0^b} } < \delta$.  
\end{itemize}
$R(U)$ is constructed uniformly in the parameters $\Theta$ and $\delta$, and its structure does not depend on~$U$.  
\end{theorem}

The phase-detection procedure can be constructed using, for example, standard phase estimation~\cite{Kitaev95stabilizer, CleveEkertMacchiavelloMosca98algorithms, NagajWocjanZhang09qma}.  Phase detection is a common subroutine in quantum algorithms, used implicitly or explicitly in, e.g., \cite{Szegedy04walkfocs, AmbainisChildsReichardtSpalekZhang07andor, MagniezNayakRolandSantha07search, ReichardtSpalek08spanprogram, MagniezNayakRichterSantha08hitting, Reichardt09spanprogram_arxivandfocs, Reichardt10advtight}.  

\smallskip

Now we are ready to construct the algorithm to prove \thmref{t:nonnon}.  Let $W = \gtwoDelta{\rho - \sigma}$.  Let $\{\ket{u_{xj}}\}$ and $\{\ket{v_{xj}}\}$, vectors in $\C^m$, be a solution to Eq.~\eqnref{e:filteredgamma2def} for $\gtwoDelta{\rho - \sigma}$.  The first step is to turn this solution into a more natural geometric object.  If the input alphabet size is $\abs D = k$, let $\ket{\u_i}, \ket{\ut_i}$ be the vectors given in \factref{t:gamma2}.  Notice that we can rewrite the sum from Eq.~\eqnref{e:filteredgamma2def} $\sum_{j \in [n]} (\Delta_j)_{x,y} \braket{u_{xj}}{v_{yj}} = \sum_{j : \, x_j \ne y_j} \braket{u_{xj}}{v_{yj}}$ as simply the inner product between the vectors $\sum_j \ket j \ket{u_{xj}} \ket{\u_{x_j}}$ and $\frac{2(k-1)}{k} \sum_j \ket j \ket{v_{yj}} \ket{\ut_{y_j}}$.  Our algorithm is based on these combined~vectors.  

Let $\H$ be the Hilbert space for the states $\ket{\rho_x}$ and $\ket{\sigma_x}$.  For $y \in \D$, let $\Pi_y = \identity - \sum_j \ketbra j j \otimes \ketbra{\u_{y_j}}{\u_{y_j}} \otimes \identity_{\C^m}$.  Also, define vectors $\ket{t_{y\pm}}, \ket{\psi_y} \in (\C^2 \otimes \H) \oplus (\C^n \otimes \C^k \otimes \C^m)$ by 
\begin{equation*}\begin{split}
\ket{t_{y\pm}} &= \tfrac{1}{\sqrt 2} \big(\ket 0 \otimes \ket{\rho_y} \pm \ket 1 \otimes \ket{\sigma_y}\big) \\
\ket{\psi_y} &= \frac{\epsilon}{\sqrt{W}} \ket{t_{y-}} - \sum_{j \in [n]} \ket j \otimes \ket{\u_{y_j}} \otimes \ket{u_{y j}}
 \enspace .
\end{split}\end{equation*}
Let $\Lambda$ be the projection onto the {orthogonal complement} of the span of the vectors $\{ \ket{\psi_y} \}_{y \in \D}$.  Then our algorithm is given~by: 

\newlength{\algboxwidth} \setlength{\algboxwidth}{5.65in}
\iftwocolumn \setlength{\algboxwidth}{.93 \columnwidth} \fi
\def\algbox#1{\begin{center}\fbox{ \begin{minipage}[l]{\algboxwidth}#1\end{minipage} }\end{center}}
\algbox{
\noindent {\bf Algorithm:} 
On input $x$, let $U_x = (2\Pi_x-\identity)(2\Lambda-\identity)$.  Apply the phase-detection circuit $R(U_x)$, from \thmref{t:phasedetection}, with precision $\Theta = \epsilon^2/W$ and error $\delta = \epsilon$, on input state $\ket 0 \otimes \ket{\rho_x} \otimes \ket{0^b}$.  Output the result.  
}

Note that the reflection $2\Pi_x-\identity$ can be computed with a query to the input oracle and its inverse---compute $x_j$, reflect in $\ket{\mu_{x_j}}$, then uncompute $x_j$.  Therefore, the algorithm uses $O(W/\epsilon^2 \cdot \log \tfrac1\epsilon)$ input queries.  The algorithm is based on repeated reflections.  As sketched in \figref{f:example}, the~algorithm can also be interpreted as a quantum walk on the bipartite graph with biadjacency matrix $\Pi_x + \sum_y \ketbra{y}{\psi_y}$.  

\begin{figure}
\centering
\includegraphics{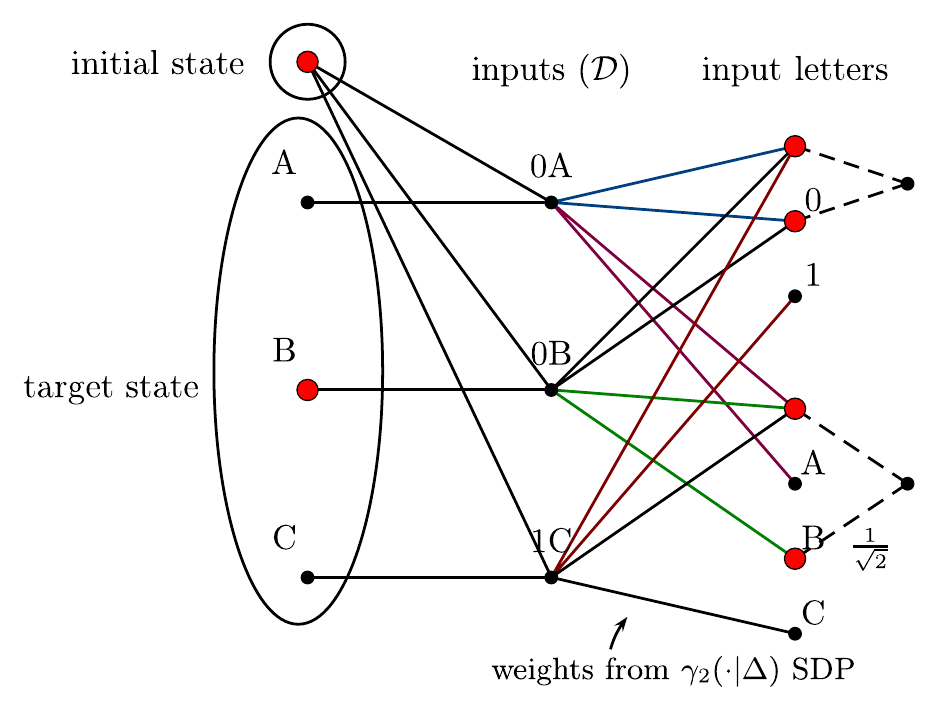}
\caption{The algorithm can be interpreted as running a quantum walk on a bipartite graph with weighted biadjacency matrix~$\Pi_x + \sum_y \ketbra{y}{\psi_y}$.  Shown above is an example graph for the function $f(x_1 x_2) = x_2$ mapping $\D = \{0A, 0B, 1C\} \subset \{0,1\} \times \{\text{A}, \text{B}, \text{C}\}$ to $E = \{\text{A}, \text{B}, \text{C}\}$.} \label{f:example}
\end{figure}

To get some intuition for why this algorithm works, observe that the initial state $\ket 0 \otimes \ket{\rho_x}$ is $\frac{1}{\sqrt 2} (\ket{t_{x+}} + \ket{t_{x-}})$.  The vector $\ket{t_{x+}}$ has large overlap with an eigenvalue-one eigenvector of~$U_x$ (\claimref{t:tplus} below), whereas the cumulative squared overlap of $\ket{t_{x-}}$ with eigenvectors of $U_x$ with small angle is small (\claimref{t:tminus}).  The phase-detection procedure therefore approximately reflects the $\ket{t_{x-}}$ term, giving roughly $\frac{1}{\sqrt 2} (\ket{t_{x+}} - \ket{t_{x-}}) = \ket 1 \otimes \ket{\sigma_x}$---our target state.  

Now we give the formal analysis of the algorithm.  Let $\{ \ket \beta \}$ be a complete set of eigenvectors of~$U_x$ with corresponding eigenvalues $e^{i \theta(\beta)}$, $\theta(\beta) \in (-\pi, \pi]$.  For an angle $\Theta \geq 0$, let $P_\Theta = \sum_{\beta: \abs{\theta(\beta)} \leq \Theta} \ketbra \beta \beta$, and $\overline P_\Theta = \identity - P_\Theta$.  

\begin{claim} \label{t:tplus}
$\norm{P_0 \ket{t_{x+}}}^2 \geq 1 - \epsilon^2$.  
\end{claim}

\begin{proof}
We give a state $\ket \varphi$ such that $U_x \ket \varphi = \ket \varphi$ and $\abs{\braket{t_{x+}}{\varphi}}^2 / \norm{\ket \varphi}^2 \ge 1 - \epsilon^2$.  Let 
\begin{equation*}
\ket \varphi = \ket{t_{x+}} + \frac12 \frac{\epsilon}{\sqrt W} \frac{2(k-1)}{k} \sum_{j \in [n]} \ket{j} \otimes \ket{\ut_{x_j}} \otimes \ket{v_{xj}}
 \enspace .
\end{equation*}
Then $\ket \varphi$ is orthogonal to all $\ket{\psi_y}$, since $\braket{t_{y-}}{t_{x+}} = \frac12 \left(\braket{\rho_y}{\rho_x} - \braket{\sigma_y}{\sigma_x} \right)$ implies 
\begin{equation*}
\braket{\psi_y}{\varphi} = \frac{\epsilon}{\sqrt W} \braket{t_{y-}}{t_{x+}} - \frac12 \frac{\epsilon}{\sqrt W} \sum_{j : x_j \neq y_j} \braket{u_{yj}}{v_{xj}} 
= 0
 \enspace .
\end{equation*}
Thus $\Lambda \ket \varphi = \ket \varphi$.  Since $\Pi_x \ket{t_{x+}} = \ket{t_{x+}}$ and $\braket{\mu_{x_j}}{\nu_{x_j}} = 0$ for all $j \in [n]$, also $\Pi_x \ket \varphi = \ket \varphi$.  
\end{proof}

\begin{claim} \label{t:tminus}
For all $\Theta \geq 0$, $\norm{P_\Theta \ket{t_{x-}}}^2 \leq \frac{\Theta^2}{4} \big( \frac{W^2}{\epsilon^2}+1 \big)$.  
\end{claim}

\begin{proof}
Apply \lemref{t:ref_lemma} with $\Pi = \Pi_x$ and $\ket w = \tfrac{\sqrt W}{\epsilon} \ket{\psi_x}$.  Then $\Delta \ket w = 0$ and $\ket{t_{x-}} = \Pi_x \ket w$.  
\end{proof}

The following proposition completes the proof of \thmref{t:nonnon}: 

\begin{proposition}
If $W \geq \epsilon$, then with the parameters $\delta = \epsilon$ and $\Theta = \epsilon^2 / W$, 
\[
\Bignorm{ R(U_x) \ket 0 \otimes \ket{\rho_x} \otimes \ket{0^b} - \ket 1 \otimes \ket{\sigma_x} \otimes \ket{0^b} } < 4 \epsilon
 \enspace .
\]
\end{proposition}

\begin{proof}
We have  
\begin{align*}
\Bignorm{ R(U_x) \ket 0 \ket{\rho_x} \ket{0^b} - \ket 1 \ket{\sigma_x} \ket{0^b} }
&= \frac{1}{\sqrt 2} \Bignorm{ R(U_x) (\ket{t_{x+}} + \ket{t_{x-}}) \ket{0^b} - (\ket{t_{x+}} - \ket{t_{x-}}) \ket{0^b} } \\
&\leq \frac{1}{\sqrt 2} \Bignorm{ (R(U_x) - \identity) \ket{t_{x+}} \ket{0^b} } + \frac{1}{\sqrt 2} \Bignorm{ (R(U_x) + \identity) \ket{t_{x-}} \ket{0^b} }
 \enspace .
\end{align*}
By \thmref{t:phasedetection}, the first term equals $\tfrac{1}{\sqrt 2} \norm{(R(U_x) - \identity) \overline P_0 \ket{t_{x+}} \ket{0^b}} \leq \sqrt 2 \norm{\overline P_0 \ket{t_{x+}}} \leq \sqrt 2 \epsilon$, by \claimref{t:tplus}.  The second term is at most $\tfrac{1}{\sqrt 2} \norm{ (R(U_x) + \identity) \overline P_\Theta \ket{t_{x-}} \ket{0^b} } + \sqrt 2 \norm{P_\Theta \ket{t_{x-}}} < \frac{\delta}{\sqrt 2} + \frac{\Theta}{\sqrt 2} \sqrt{\frac{W^2}{\epsilon^2}+1}$, by \thmref{t:phasedetection} and \claimref{t:tminus}.  Now substitute our choices of parameters and use $\frac{W^2}{\epsilon^2} + 1 \leq 2 W^2 / \epsilon^2$ to conclude the proof.  
\end{proof}

Notice that the constant factor hidden by the big-$O$ notation in \thmref{t:nonnon} is the same as the constant hidden in \thmref{t:phasedetection} for the number of calls to $U$ and $U^{-1}$, and is less than~$100$.

\subsection{Lower bound for state conversion}

We now show how $\gtwoDelta{\rho - \sigma}$ can be used to show query complexity lower bounds for the state-conversion problem.  The argument has two parts.  First, we show that $\gtwoDelta{\rho - \sigma}$ lower-bounds~the complexity of {\em exactly} converting $\rho$ to $\sigma$.  Second, we develop an output condition constraining those~$\sigma'$ that are viable final Gram matrices of a successful algorithm with error~$\epsilon$.  These two~parts have been present in all previous adversary arguments, but the separation has not been fully~recognized.  

Once these two steps are finished, the lower bound naturally becomes $\min_{\sigma' \approx_\epsilon \sigma} \gtwoDelta{\rho-\sigma'}$, where the notion of approximation is given by the output condition.  This paradigm follows the use of approximation norms for lower bounds in communication complexity~\cite{LeeShraibman07survey}.  

We begin with the lower bound for exact state conversion: 

\begin{lemma} \label{t:gamma2lowerbound}
Suppose that $\sigma$ can be reached from $\rho$ with one query.  Then $\gtwoDelta{\rho - \sigma} \le 2$.  
\end{lemma}

\begin{proof}
Let $\Gamma_j$ project onto the query register containing index~$j$.  For $x \in \D$, let $O_x$ be the unitary query oracle.  It satisfies $O_x^\dagger O_y \Gamma_j = \Gamma_j$ when $x_j = y_j$.  By assumption, $\sigma_{x,y} = \bra{\rho_x} O_x^\dagger O_y \ket{\rho_y}$.  Then 
\begin{equation*}
(\rho-\sigma)_{x,y}
= \sum_j \bra{\rho_x}\Gamma_j\ket{\rho_y} - \bra{\rho_x}O_x^\dagger O_y\Gamma_j\ket{\rho_y} 
= \sum_{\substack{j : x_j \ne y_j}} \bra{\rho_x}\Gamma_j\ket{\rho_y} - \bra{\rho_x}O_x^\dagger O_y\Gamma_j\ket{\rho_y}
 \enspace .
\end{equation*}
Now define $\ket{u_{xj}} = (\Gamma_j \ket{\rho_x}, O_x \Gamma_j \ket{\rho_x})$ and $\ket{v_{xj}} = (\Gamma_j \ket{\rho_x}, -O_x \Gamma_j \ket{\rho_x})$.  
Then $\braket{u_{xj}}{v_{yj}} = \bra{\rho_x} (\identity - O_x^\dagger O_y) \Gamma_j \ket{\rho_y}$, as desired.  Furthermore, $\sum_j \norm{\ket{u_{xj}}}{}^2 = \sum_j \norm{\ket{v_{xj}}}{}^2 = 2$.  
\end{proof}

Lower bounds for approximate query problems follow by combining this lemma with appropriate output conditions.  For example, in the functional case, one can use a condition based on $\ell_\infty$ distance~\cite{Ambainis00adversary}, or the full output condition from~\cite{BarnumSaksSzegedy03adv}.  The output condition traditionally used for the general adversary method is based on the $\gamma_2$ norm~\cite{HoyerLeeSpalek07negativeadv}.  This condition has the advantage that it is an SDP, it extends to state conversion, and, as we now show, it is tight.  

\begin{lemma} \label{t:gamma2outputconditiontight}
Let $\{ \ket{\rho_x} \}, \{ \ket{\sigma_x} \} \subset \H$ be finite sets of vectors with the same index set, and let~$\rho, \sigma$ be their respective Gram matrices.  Then 
\begin{itemize}
\item
If $\Re{(\braket{\rho_x}{\sigma_x})}  \geq \sqrt{1 - \epsilon}$ for every $x$, then $\gamma_2(\rho - \sigma) \leq 2 \sqrt \epsilon$~\cite{HoyerLeeSpalek07negativeadv}.  
\item
If $\gamma_2(\rho - \sigma) \leq \epsilon$, then there exists a unitary~$U$ such that $\bra{\rho_x} U \ket{\sigma_x} \geq 1 - \sqrt{2 \epsilon}$ for all~$x$.  
\end{itemize}
\end{lemma}
The second item has recently been improved by \cite{LeeRoland11directproduct} to $\gamma_2(\rho - \sigma) \leq \epsilon$ implies there exists a unitary~$U$ such that $\bra{\rho_x} U \ket{\sigma_x} \geq 1 - \epsilon/2$ for all~$x$.  

\begin{proof}
For the first part of the lemma, we can factorize $\rho-\sigma$ as  
\[ 
(\rho-\sigma)_{x,y} = \tfrac12 \left(\braket{\rho_x+\sigma_x}{\rho_y - \sigma_y} + \braket{\rho_x-\sigma_x}{\rho_y + \sigma_y}\right)
 \enspace .
\]  
Thus by a triangle inequality $\gamma_2(\rho-\sigma) \leq \max_{x,y} \ \norm{\rho_x+\sigma_x}\norm{\rho_y-\sigma_y} \le 2 \max_y \ \sqrt{2-2\Re{(\braket{\rho_y}{\sigma_y}})} \le 2 \sqrt \epsilon$.  For the last inequality we used $\sqrt{1-\epsilon} \ge 1- \tfrac\epsilon2$.  

To prove the second part, let $\{u_x\}$ and $\{v_x\}$ be arbitrary factorizations of $\rho$ and $\sigma$, respectively.  
As $\gamma_2(\rho-\sigma) \le \epsilon$, there exists a factorization $(\rho-\sigma)_{x,y} = \braket{\alpha_x}{\beta_y}$ with $\norm{\alpha_x},\norm{\beta_y} \le\sqrt{\epsilon}$.  Then  
\begin{align*}
\braket{v_x}{v_y} &= \braket{u_x}{u_y} - \braket{\alpha_x}{\beta_y}\\
&= \braket{u_x}{u_y} - \frac{1}{2}\braket{\alpha_x}{\beta_y} - \frac{1}{2} \braket{\beta_x}{\alpha_y} 
\end{align*}
as $\rho - \sigma$ is Hermitian.  Let $p_x = \frac{1}{2} (\alpha_x - \beta_x)$ and $q_x = \frac{1}{2} (\alpha_x + \beta_x)$.  Then the previous equation implies
\[
\braket{(u_x,p_x)}{(u_y,p_y)} =  \braket{(v_x,q_x)}{(v_y,q_y)}
 \enspace .
\] 
By unitary freedom of square roots, if $AA^\dagger = BB^\dagger$ for any two matrices $A$ and $B$, then $AU = B$ for some unitary $U$.  So there is a unitary $U$ such that $(u_x,p_x) U = (v_x,q_x)$.  As $\braket{\alpha_x}{\beta_x}=0$ because $\rho -\sigma$ has zeros on the diagonal, we have $\norm{p_x}^2=\norm{q_x}^2 \le \epsilon/2$.  
\begin{align*}
\braket{(u_x,0)U}{(v_x,0)} 
&= \braket{(u_x, p_x)U}{(v_x, q_x)} - \braket{(u_x, 0)U}{(0,q_x)} - \braket{(0,p_x)U}{(v_x, 0)} - \braket{(0,p_x)U}{(0,q_x)} \hspace{-.2in} \\
&\ge 1 + \norm{q_x}^2 - \norm{p_x} - \norm{q_x} - \norm{p_x}\norm{q_x} \\
&\ge 1 - \sqrt{2\epsilon}
 \enspace .  \qedhere
\end{align*}
\end{proof}

Based on \lemref{t:gamma2outputconditiontight}, we immediately derive tight SDPs for the query complexities of state conversion:  

\begin{theorem} \label{t:nonnonnon}
For $\delta > 0$, let 
\begin{equation}\begin{split}
q_\delta(\rho, \sigma) &= \;\min_{\sigma' \succeq 0} \;\, \Big\{ \gtwoDelta{\rho - \sigma'} : \gamma_2(\sigma' - \sigma) \le \delta \Big\} \\
q_\delta^{nc}(\rho, \sigma) &= \min_{\sigma', S \succeq 0} \Big\{ \gtwoDelta{\rho - \sigma'} : \gamma_2(\sigma' - \sigma \circ S) \le \delta, \; S \circ \identity = \identity \Big\} 
 \enspace .
\end{split}\end{equation}
Then the bounded-error coherent and non-coherent state-conversion query complexities satisfy 
\begin{equation}\begin{split}
\Omega\Big( q_{2 \sqrt{2 \epsilon}}(\rho, \sigma) \Big) \leq Q_\epsilon(\rho, \sigma) \leq O\Big(q_{\epsilon^4/16}(\rho, \sigma) \frac{\log(1/\epsilon)}{\epsilon^2}\Big) \\
\Omega\Big( q_{2 \sqrt{2 \epsilon}}^{nc}(\rho, \sigma) \Big) \leq Q_\epsilon^{nc}(\rho, \sigma) \leq O\Big(q_{\epsilon^4/16}^{nc}(\rho, \sigma) \frac{\log(1/\epsilon)}{\epsilon^2}\Big) 
 \enspace .
\end{split}\end{equation}
\end{theorem}

Thus for coherent state conversion, the output condition used is $\gamma_2(\sigma' - \sigma) \le \delta$ for an appropriately chosen~$\delta$, and in the non-coherent case, optimization over additional garbage states is allowed.  

For well-behaved problems, i.e., problems satisfying $Q_{1/3}(\rho,\sigma) = O\big(Q_\epsilon(\rho,\sigma) \log(1/\epsilon)\big)$ as in the functional case, this is true characterization.  General state-conversion problems, however, do not necessarily satisfy this robustness condition.  Just as the complexity of a boolean function can have a precipitous change around error $1/2$, state-conversion problems can have non-continuous changes in complexity even around small values of $\epsilon$.  For such problems \thmref{t:nonnonnon} may not be a true characterization because of the gap in error parameters on the left- and right-hand sides.  The gap in the error dependence arises from the looseness of the necessary and sufficient conditions in \lemref{t:gamma2outputconditiontight}, plus the error from \thmref{t:nonnon}.  We do not know if the $\epsilon$-dependence in \thmref{t:nonnon} can be improved to polylogarithmic in $1/\epsilon$.  

An advantage of using the $\gamma_2$ output condition is that the quantities $q_\delta$ and $q_\delta^{nc}$ are described by semi-definite programs.  One could define analogous quantities with other output conditions, however, including the ``true'' output condition given by \defref{def_coherent} and \defref{def_nc}.  In this case the only slack in the characterization would arise from \thmref{t:nonnon} and thus the error parameters on left- and right-hand sides would agree up to constant factors.  

Ambainis et al.~\cite{AmbainisMagninRoettelerRoland10stategeneration}, previously extended both the general adversary bound and the multiplicative adversary bound \cite{Spalek07multiplicative} to the state-generation problem.  A difference between our work and theirs is that we separate the bound for the exact problem from the output condition used to handle the bounded error case.  \cite{AmbainisMagninRoettelerRoland10stategeneration} focus on the output condition introduced by \cite{Spalek07multiplicative} with the multiplicative adversary method and show how to extend the additive adversary method with this output condition to the state generation problem, calling this the hybrid adversary method.  They show that the hybrid adversary method dominates the general adversary method, and that the hybrid adversary method is dominated by the bound of~\cite{Spalek07multiplicative} extended to the case of state generation.  We do not know if the hybrid adversary method also dominates the $q_\delta(J, \sigma)$ measure.  The proof in~\cite{AmbainisMagninRoettelerRoland10stategeneration} that the multiplicative method dominates the hybrid method actually shows that the multiplicative method for exact state generation dominates the $\gamma_2(J-\sigma | \Delta)$ measure, as explicitly shown by~\cite{LeeRoland11directproduct}.  Thus the multiplicative method will dominate the $\gamma_2(J -\sigma | \Delta)$ bound whenever they are paired with the same output condition.  

This line of research into discrete query complexity was launched by the discovery of a continuous-time query algorithm for evaluating AND-OR formulas~\cite{fgg:and-or}.  We now complete the circle: 

\begin{theorem} \label{t:continuousdiscrete}
The bounded-error continuous-time and discrete query models are equivalent.  
\end{theorem}

Cleve et al.~\cite{CleveGottesmanMoscaSommaYongeMallo08discretize} have shown that the models are equivalent up to a sub-logarithmic factor.  The proof of \thmref{t:continuousdiscrete}, given in \appref{a:continuousquery}, follows from our algorithm, \thmref{t:nonnon}, together with the observation that the general adversary bound remains a lower bound for continuous-time query algorithms.  The latter result has been observed by Yonge-Mallo in 2007 \cite{YM11} and, independently, Landahl (personal communication).

\section{Function composition}

In this section we show that the adversary method behaves well with respect to function composition, extending previous work for the boolean case~\cite{HoyerLeeSpalek07negativeadv, Reichardt09spanprogram_arxivandfocs}.  Let $g: {\mathcal C} \rightarrow D$ where ${\mathcal C} \subseteq C^m$ and $f : \D \rightarrow E ~~ (\D \subseteq D^n)$ for finite sets $C,D,E$.  Define the composed function $f \circ g^n$ by 
\begin{equation*}
(f \circ g^n)(x) = f\big(g(x_1, \ldots, x_m), \ldots, g(x_{(n-1)m+1}, \ldots, x_{mn})\big)
 \enspace .
\end{equation*}

\begin{lemma} \label{t:ADVpmwsizecomposition}
Letting $G = \{\delta_{g(x),g(y)}\}_{x,y}$, $\ADVpm(f \circ g^n) \le \ADVpm(f) \, \gtwoDelta{J - G}$.  
\end{lemma}

The proof of this lemma follows in the natural way.  We take optimal solutions to the $\ADVpm(f)$ program and the $\gtwoDelta{J-G}$ program, and form their tensor product to construct a solution to the composed program.  This proof strategy does not directly work for the general adversary bound---we crucially use the extra constraints present in the $\gtwoDelta{J-G}$ program.  In fact, this lemma can be seen as a special case of the more general inequality $\gtwo{A}{Z} \leq \gtwo{A}{Y} \max_j \gtwo{Y_j}{Z}$ (\lemref{t:gamma2_properties}).  The proof is given in \appref{a:functioncomposition}.  

In the case where all the functions $f$ and $g$ have boolean inputs and outputs, a matching lower bound to \lemref{t:ADVpmwsizecomposition} has been shown by H{\o}yer et al.~\cite{HoyerLeeSpalek07negativeadv}.  In general, we cannot always show such a matching lower bound.  For example, let $g$ be a function that only outputs even numbers, and let $f$ output the sum of its inputs modulo two; then $f \circ g^n$ is constant.  We can, however, show a matching composition lower bound when the range of~$g$ is boolean: 

\begin{lemma} \label{t:compositionlowerbound}
Let $g: {\mathcal C} \rightarrow \B$ and $f : \B^n \rightarrow E$.  Then $\ADVpm(f \circ g^n) \ge \ADVpm(f) \ADVpm(g)$.  
\end{lemma}

The proof is given in \appref{a:functioncomposition}.  The above composition lemmas also lead to direct-sum results for quantum query complexity: 

\begin{corollary} \label{t:directsumcorollary}
Let $g : \D \rightarrow E$, and let $g^n : \D^n \rightarrow E^n$ consist of $n$ independent copies of $g$, given by $g^n(x^1, \ldots, x^n) = \big(g(x^1), \ldots, g(x^n)\big)$.  Then 
\begin{equation}
Q(g^n) = \Theta\big(n \, Q(g)\big)
 \enspace .
\end{equation}
\end{corollary}

The lower bound $\ADVpm(g^n) \geq n \, \ADVpm(g)$ has been shown by~\cite{AmbainisChildsLegallTani09witness}.  The corresponding upper bound $\ADVpm(g^n) \leq n \, \ADVpm(g)$ is a special case of \lemref{t:ADVpmwsizecomposition}, with $f$ the identity function.  \corref{t:directsumcorollary} then follows from \thmref{t:advtight}.  Let us remark that when $E = \{0,1\}$, the upper bound $Q(f^n) = O\big(n \, Q(f)\big)$ follows from the robust input recovery quantum algorithm~\cite[Theorem~3]{BuhrmanNewmanRohrigDeWolf05robust}.  The same algorithm can be generalized to handle larger $E$.

\subsection*{Acknowledgements}

We thank J\'{e}r\'{e}mie Roland and Miklos Santha for many helpful conversations on these topics, and thank Richard Cleve and Ronald de Wolf for useful comments on an earlier draft.  Part of this work was done while T.L.~was at Rutgers University, supported by an NSF postdoctoral fellowship and grant CCF-0728937.  R.M.~acknowledges support from NSERC and NSF grant CCF-0832787.  B.R.~acknowledges support from NSERC, ARO-DTO and MITACS.

\bibliographystyle{alpha-eprint}
\bibliography{andor}

\bigskip 

\newcommand{\address}[2]{{\noindent \small {\sc{#1}} \\ \emph{E-mail address:} {\tt{#2}} \medskip}}

\address{Centre for Quantum Technologies}{troyjlee@gmail.com}

\address{Institute for Quantum Computing, University of Waterloo}{r3mittal@uwaterloo.ca}

\address{Institute for Quantum Computing, University of Waterloo}{breic@iqc.ca}

\address{Google, Inc.}{spalek@google.com}

\address{Rutgers University}{szegedy@cs.rutgers.edu}

\appendix

\section{Properties of the filtered \texorpdfstring{$\gamma_2$}{gamma\_2} norm} \label{s:gamma2properties}

For reference, we list several useful properties of the filtered $\gamma_2$ norm.  First, we give an alternative formulation that explains why $\gamma_2$ is also known as the Schur product operator norm: 

\begin{lemma}
The $\gamma_2$ and filtered $\gamma_2$ norms can equivalently be expressed by 
\begin{align} 
\label{e:gamma2dual}
\gamma_2(A) &= \max_M \big\{ \norm{A \circ M} : \norm{M} \leq 1 \big\} \\
\label{e:filteredgamma2dual}
\gtwo{A}{Z} &= \max_M \big\{ \norm{A \circ M} : \max_j \norm{Z_j \circ M} \leq 1 \big\}
 \enspace .
\end{align}
\end{lemma}

\begin{proof}
Both of these equations can be proven in the same way: Start with Eq.~\eqnref{e:gamma2def} or~\eqnref{e:filteredgamma2def} for $\gamma_2(A)$ or $\gtwo{A}{Z}$, respectively,  and take the dual.  The semi-definite program given by Eq.~\eqnref{e:gamma2def} is always strictly feasible and that of Eq.~\eqnref{e:filteredgamma2def} is strictly feasible provided that whenever $A_{x,y} \ne 0$ there is some $j$ with $(Z_j)_{x,y} \ne 0$, i.e., provided the maximum in \eqnref{e:filteredgamma2dual} is finite.  Thus by the duality principle~\cite[Theorem~3.4]{Lovasz00sdp} the primal and dual formulations are equal and the optimum is achieved.  
\end{proof}

\begin{lemma} \label{t:gamma2_properties}
The quantity $\gtwo{\cdot}{Z}$ is a norm when restricted to arguments supported on the union 
of the supports of the~$Z_j$.  For matrices $B, Y_1, \ldots, Y_n$ of the appropriate dimensions it satisfies: 
\begin{enumerate}
\item \label{e:propertyspecialcases} If $A \neq 0$, then $\gtwo{A}{\{A\}} = 1$.  For $J$ the all-ones matrix, $\gtwo{A}{\{J\}} = \gamma_2(A)$.  
\item \label{e:propertyextremes} $\gtwo{A}{Z}=0$ if and only if $A=0$.  $\gtwo{A}{Z} = \infty$, i.e., Eq.~\eqnref{e:filteredgamma2def} is infeasible, if and only if there exists an entry $(x,y)$ such that $A_{x,y} \neq 0$ and $(Z_j)_{x,y} = 0$ for all $j$.  
\item \label{e:propertyscalable} Positive scalability: $\gtwo{s A}{Z} = \abs{s} \gtwo{A}{Z}$ and $\gtwo{A}{\{s Z_1, \ldots, s Z_n\}} = \frac{1}{\abs{s}}\gtwo{A}{Z}$ for $s \neq 0$.  
\item \label{e:propertytriangle}Triangle inequality: $\gtwo{A+B}{Z} \le \gtwo{A}{Z} + \gtwo{B}{Z}$.  
\item \label{e:propertyduplication} $\gtwo{A}{Z}$ is invariant under duplicating corresponding rows or columns of $A$ and all $Z_j$.  
\item \label{e:propertyunion} $\gtwo{A}{Y \cup Z} \leq \gtwo{A}{Z}$.  This is an equality if each $Y_i$ is a restriction of some $Z_j$ to a rectangular submatrix.  
\item \label{e:propertyconvexhull} Provided $\sum_j \abs{p_j} = 1$, $\gtwo{A}{Z} = \gtwo{A}{Z \cup \{\sum_j p_j Z_j\}}$.  (Thus the second argument in $\gtwo{A}{Z}$ can be thought of as a convex set centered at the origin, where only the extremal points matter.)   
\item \label{e:propertyrowcoldisjoint} If the supports of $Z_1$ and $Z_2$ are row- and column-disjoint, then $\gtwo{A}{Z} = \gtwo{A}{\{Z_1 + Z_2, Z_3, \ldots, Z_n\}}$.  
\item \label{e:propertyfilteredgamma2gamma2} $\gtwo{A \circ B}{Z} \leq \gtwo{A}{Z} \gamma_2(B)$.  
\item \label{e:property6} $\gtwo{A \circ B}{\{Z_j \circ B\}} \leq \gtwo{A}{Z} \leq \gtwo{A}{\{Z_j \circ B\}} \gamma_2(B)$.  
\item \label{e:propertycomposition} A composition property: $\gtwo{A}{Z} \leq \gtwo{A}{Y} \max_j \gtwo{Y_j}{Z}$.  
\item A direct-sum property: $\gtwo{A \oplus B}{\{Y_j \oplus Z_j\}} = \max\{\gtwo{A}{Y}, \gtwo{B}{Z}\}$.  
\item \label{e:propertytensorproduct} A tensor-product property: $\gtwo{A \otimes B}{Y \otimes Z} = \gtwo{A}{Y} \gtwo{B}{Z}$, where $Y \otimes Z = \{Y_i \otimes Z_j\}$, all pairwise tensor products.  
\end{enumerate}
\end{lemma}

\begin{proof}
By items~\eqnref{e:propertyextremes}, \eqnref{e:propertyscalable} and \eqnref{e:propertytriangle}, $\gtwo{\cdot}{Z}$ is a norm on arguments restricted to the support of the $Z_j$.  The proofs of the first three properties follow easily from the definition of filtered $\gamma_2$ norm, Eq.~\eqnref{e:filteredgamma2def}.  Therefore we begin by showing the triangle inequality, property~\eqnref{e:propertytriangle}.  
\begin{enumerate}
\setcounter{enumi}{3}
\item Given optimal vector solutions to Eq.~\eqnref{e:filteredgamma2def} for $\gtwo{A}{Z}$ and for $\gtwo{B}{Z}$, simply concatenate corresponding vectors to obtain a solution for $\gtwo{A+B}{Z}$, with objective value at most $\gtwo{A}{Z} + \gtwo{B}{Z}$.  
\item Invariance of $\gtwo{A}{Z}$ under copying rows  follows by copying the associated solution vectors.  
\item Any solution to Eq.~\eqnref{e:filteredgamma2dual} for $\gtwo{A}{Y \cup Z}$ also works for $\gtwo{A}{Z}$; hence $\gtwo{A}{Z} \geq \gtwo{A}{Y \cup Z}$.  However, if $Y_i$ is a submatrix restriction of $Z_j$ then the constraint $\norm{Y_i \circ M} \leq 1$ is redundant to $\norm{Z_j \circ M} \leq 1$; hence adding $Y_i$ to~$Z$ does not affect $\gtwo{A}{Z}$.  
\item If $\max_j \norm{Z_j \circ M} \leq 1$, then $\bignorm{\sum_j p_j Z_j \circ M} \leq 1$; again, the new constraint is redundant.  
\item Assuming without loss of generality that in a solution to $\gtwo{A}{Z}$ the vectors $\ket{u_{xj}}$ (respectively, $\ket{v_{yj}}$) are nonzero only on rows (columns) where $Z_j$ has nonzero entries, concatenating the vectors for $Z_1$ and for $Z_2$ gives a solution to $\gtwo{A}{\{Z_1 + Z_2, Z_3, \ldots, Z_n\}}$.  Thus $\gtwo{A}{\{Z_1 + Z_2, Z_3, \ldots, Z_n\}} \leq \gtwo{A}{Z}$.  For the other direction, divide the vectors for $Z_1 + Z_2$ according to whether they correspond to a nontrivial row or column of $Z_1$, or of $Z_2$.  
\item Begin with an optimal solution $\{\ket{u_x}, \ket{v_y}\}$ to Eq.~\eqnref{e:gamma2def} for $\gamma_2(B)$, and an optimal solution $\{\ket{u_{xj}}, \ket{v_{yj}}\}$ to Eq.~\eqnref{e:filteredgamma2def} for $\gtwo{A}{Z}$.  The tensor products $\{ \ket{u_{xj}} \otimes \ket{u_x}, \ket{v_{yj}} \otimes \ket{v_y}\}$ give a solution for $\gtwo{A \circ B}{Z}$, with objective value at most $\gtwo{A}{Z} \gamma_2(B)$.  

\item The second inequality in property~\eqnref{e:property6} works in the same way as~\eqnref{e:propertyfilteredgamma2gamma2}; the tensor product of vector solutions for $\gamma_2(B)$ and 
$\gtwo{A}{\{Z_j \circ B\}}$ is a solution for $\gtwo{A}{Z}$.  The first inequality follows since any vector solution for $\gtwo{A}{Z}$ also works for $\gtwo{A \circ B}{\{Z_j \circ B\}}$.  

\item Let $\{\ket{u_{xj}}, \ket{v_{yj}}\}$ be an optimal solution to Eq.~\eqnref{e:filteredgamma2def} for $\gtwo{A}{Y}$ and for each $j$ let $\{\ket{u_{xi}^j}, \ket{v_{yi}^j}\}$ be an optimal solution for $\gtwo{Y_j}{Z}$.  These vectors satisfy 
\begin{align*}
\gtwo{A}{Y} &\geq \max\Big\{ \sum_j \norm{\ket{u_{xj}}}{}^2, \sum_j \norm{\ket{v_{yj}}}{}^2 \Big\} &
\gtwo{Y_j}{Z} &\geq \max\Big\{ \sum_i \norm{\ket{u_{xi}^j}}{}^2, \sum_i \norm{\ket{v_{yi}^j}}{}^2 \Big\} \\
A_{x,y} &= \sum_j (Y_j)_{x,y} \braket{u_{xj}}{v_{yj}} &
(Y_j)_{x,y} &= \sum_i (Z_i)_{x,y} \braket{u_{xi}^j}{v_{yi}^j}
 \enspace .
\end{align*}
Combining the last two equations gives $A_{x,y} = \sum_i (Z_i)_{x,y} \sum_j \braket{u_{xj}}{v_{yj}} \braket{u_{xi}^j}{v_{yi}^j}$.  Thus the vectors $\oplus_j ( \ket{u_{xj}} \otimes \ket{u_{xi}^j} )$ and $\oplus_j ( \ket{v_{yj}} \otimes \ket{v_{yi}^j} )$ give a solution for $\gtwo{A}{Z}$, with objective value at most $\max\{ \max_x \sum_{i,j} \norm{\ket{u_{xj}}}{}^2 \norm{\ket{u_{xi}^j}}{}^2, \max_y \sum_{i,j} \norm{\ket{v_{yj}}}{}^2 \norm{\ket{v_{yi}^j}}{}^2 \} \leq \gtwo{A}{Y} \max_j \gtwo{Y_j}{Z}$.  

\item A union of the vectors for $\gtwo{A}{Y}$ and $\gtwo{B}{Z}$ gives a vector solution for $\gtwo{A\oplus B}{Y_j \oplus Z_j}$.  

\item The inequality $\gtwo{A \otimes B}{Y \otimes Z} \leq \gtwo{A}{Y} \gtwo{B}{Z}$ is straightforward; if $\{\ket{u_{xi}}, \ket{v_{yi}}\}$ form an optimal vector solution for $\gtwo{A}{Y}$ and $\{\ket{\mu_{\alpha j}}, \ket{\nu_{\beta j}} \}$ form an optimal vector solution for $\gtwo{B}{Z}$, then the vectors $\ket{u_{(x,\alpha) (i,j)}} = \ket{u_{xi}} \otimes \ket{\mu_{\alpha j}}$ and $\ket{v_{(y,\beta) (i,j)}} = \ket{v_{yi}} \otimes \ket{\nu_{\beta j}}$ satisfy 
\begin{align*}
\sum_{i,j} (Y_i \otimes Z_j)_{(x, \alpha), (y, \beta)} \braket{u_{(x,\alpha) (i,j)}}{v_{(y,\beta) (i,j)}}
&= \sum_i (Y_i)_{x,y} \braket{u_{xi}}{v_{yi}} \sum_j (Z_j)_{\alpha, \beta} \braket{\mu_{\alpha j}}{\nu_{\beta j}} = A_{x,y} B_{\alpha, \beta} 
\end{align*}
and therefore give a solution for $\gtwo{A \otimes B}{Y \otimes Z}$, with objective value at most $\gtwo{A}{Y} \gtwo{B}{Z}$.  

For the other direction of the tensor-product inequality, let $M$ be an optimal solution to the dual SDP Eq.~\eqnref{e:filteredgamma2dual} for $\gtwo{A}{Y}$ and let $N$ be an optimal solution to the dual SDP for $\gtwo{B}{Z}$.  We claim that $M \otimes N$ is a solution to the dual SDP for $\gtwo{A \otimes B}{Y \otimes Z}$.  Indeed, for all $i, j$, $\norm{(Y_i \otimes Z_j) \circ (M \otimes N)} = \norm{(Y_i \circ M) \otimes (Z_j \circ N)} = \norm{Y_i \circ M} \norm{Z_j \circ N} \leq 1$.  The objective value is $\norm{(A \otimes B) \circ (M \otimes N)} = \gtwo{A}{Y} \gtwo{B}{Z}$.  Thus $\gtwo{A}{Y} \gtwo{B}{Z} \leq \gtwo{A \otimes B}{Y \otimes Z}$.  \qedhere
\end{enumerate}
\end{proof}

For completeness, we also present the dual norm $\gtwostar{\,\cdot\,}{Z}$.  Let $\hat A = \big[\begin{smallmatrix} 0 & A \\ A^\dagger & 0 \end{smallmatrix}\big]$.  Then 
\begin{equation}\begin{split} \label{e:gtwostar}
\gtwostar{A}{Z}
&= \max_{B : \gtwo{B}{Z} = 1} \langle A, B \rangle \\
&= \max_{\{Y_j \succeq 0\}} \Big\{ \frac12 \sum_j \langle Y_j, \hat Z_j \circ \hat A \rangle : \sum_j Y_j \circ \identity = \identity \Big\} \\ 
&= \min_\Omega \Big\{ \tfrac12 \Tr \, \Omega : \text{$\Omega \circ \identity = \Omega$ and $\forall j$, $\Omega - \hat A \circ \hat Z_j \succeq 0$} \Big\} \enspace .
\end{split}\end{equation}
When $A$ and the $Z_j$ are Hermitian, then $\gtwostar{A}{Z} = \min_\Omega \{ \Tr \, \Omega : \text{$\Omega \circ \identity = \Omega$ and $\forall j$, $\Omega \pm A \circ Z_j \succeq 0$} \}$, a slightly simpler form that we will use below in \appref{a:functioncomposition}.  The dual norm $\gamma_2^*$ satisfies several similar properties to $\gamma_2$, such as $\gtwostar{A}{Z} \leq \gtwostar{A}{Y \cup Z}$, $\gtwostar{A}{Z} = \gtwostar{A}{Z \cup \{\sum_j p_j Z_j\}}$ if $\sum_j \abs{p_j} = 1$, and $\gtwostar{A \otimes B}{Y \otimes Z} = \gtwostar{A}{Y} \gtwostar{B}{Z}$.  It also satisfies $\gtwostar{A \circ B}{Z} = \gtwostar{A}{Z \circ B} \le \gtwostar{A}{Z} \gamma_2(B)$.  We leave the proofs of these claims to the reader.

\section{Application to continuous-time query complexity} \label{a:continuousquery}

The first step of the proof of Cleve et al.~\cite{CleveGottesmanMoscaSommaYongeMallo08discretize} is to show that the continuous-time model is equivalent to the fractional quantum query model, up to constant factors.  For completeness, we now show that $\gtwoDelta{\sigma - \rho}$ remains a lower bound on the fractional query complexity, up to a constant.  Together with our upper bound, this gives \thmref{t:continuousdiscrete}.  Yonge-Mallo \cite{YM11} recently published a proof from 2007 which directly shows the general adversary bound is a lower bound on the continuous-time query model, and this was also independently observed by Landahl (unpublished).  

Let us first describe the fractional query model.  For simplicity, we restrict to the case of boolean input.  Here the $\lambda$-fractional query operator $O_x(\lambda)$ behaves as $O_x(\lambda) \ket{i}\ket{z} = e^{i \lambda \pi x_i} \ket{i}\ket{z}$.  Thus the usual query operator is obtained with $\lambda = 1$.  The query cost is $\lambda$ times the number of applications of~$O_x$.  

As before the key step is to bound how much a single query can change the distance.  

\begin{lemma}
Suppose that $\sigma$ can be reached from $\rho$ with one $\lambda$-fractional query.  Then $\gtwoDelta{\rho - \sigma} \le \lambda \pi \sqrt 2$.  
\end{lemma}

\begin{proof}
Let $\rho_x^i$ be the projection of $\rho_x$ onto the part of the query register holding $i$.  Then 
\begin{align*}
(\rho - \sigma)_{x,y}
&= \sum_{j=1}^n \big( \braket{\rho_x^j}{\rho_y^j} - \bra{\rho_x^j} e^{-i \lambda \pi x_j} e^{i \lambda \pi y_j} \ket{\rho_y^j} \big) \\
&= \sum_{j: x_j \ne y_j} \braket{\rho_x^j}{\rho_y^j}(1- e^{i \lambda \pi (y_j-x_j)}) \\
&= \sum_{j: x_j \ne y_j} \braket{\rho_x^j}{\rho_y^j} 
\big((1 - \cos(\lambda \pi)) + i (x_j - y_j) \sin(\lambda \pi) \big)
 \enspace .
\end{align*}
Now we define positive semi-definite matrices $\{P_j\}_{j \in [n]}$ satisfying $\rho-\sigma = \sum_j P_j \circ \Delta_j$.  For this let us define a couple of auxiliary matrices.  Let $M_j(x,y) = \braket{\rho_x^j}{\rho_y^j}$ and let $E_j(x,y) = \braket{e_{x_j}}{e_{y_j}}$, where $e_b = b + i (1-b)$ for $b \in \B$.  From this definition it is clear that $E_j$ is positive semi-definite, and note that 
\[
E_j(x,y) = 
\begin{cases}
i (x_j -y_j) & \text{if } x_j \ne y_j \\
1 & \text{otherwise} \enspace .
\end{cases} 
\]
Finally, we can define $P_j = (1-\cos(\lambda \pi))M_j + \sin(\lambda \pi) M_j \circ E_j$.  Then $P_j$ is positive semi-definite, and satisfies $\rho-\sigma = \sum_j P_j \circ \Delta_j$.  As $\sum_{j \in [n]} M_j(x,x)=1$ for all $x$ we can upper bound the cost $\max_x \sum_{j \in [n]} P_j(x,x)$ by $p(\lambda) = (1-\cos(\lambda \pi)) + \sin(\lambda \pi)$.  Note that $p(0) = 0$ and the maximum value of the derivative of $p(\lambda)$ is $\pi\sqrt{2}$.  Thus for $\lambda \ge 0$ we have $p(\lambda) \le \lambda \pi \sqrt 2$.  
\end{proof}

\section{Function composition} \label{a:functioncomposition}

In this section we prove the composition lemmas, Lemmas~\ref{t:ADVpmwsizecomposition} and~\ref{t:compositionlowerbound}.  

We begin with some notation.  Let $g: {\mathcal C} \rightarrow D$ where ${\mathcal C} \subseteq C^m$ and $f : D^n \rightarrow E$ for finite sets $C$, $D$ and~$E$.  Let $G = \{\delta_{g(x), g(y)}\}_{x,y}$, and $F = \{\delta_{f(x), f(y)}\}_{x,y}$.  For a string $x \in {\mathcal C}^n$ we write $x=(x^1, \ldots, x^n)$ where each $x^i \in {\mathcal C}$, and we let $\tilde x = g(x^1) \cdots g(x^n) \in D^n$.  For a ${\abs D}^n$-by-${\abs D}^n$ matrix~$A$, define a $\abs{\mathcal C}^n$-by-$\abs{\mathcal C}^n$ matrix $\tilde A$ by $\tilde A_{x,y} = A_{\tilde x, \tilde y}$.  With this notation, $\tilde \Delta_p = J^{\otimes (p-1)} \otimes (J-G) \otimes J^{\otimes (n-p)}$, and the filtering matrices for the composed function $f \circ g^n$ are $\Delta_{(p,q)} = J^{\otimes (p-1)} \otimes \Delta_q \otimes J^{\otimes (n-p)}$.  To shorten expressions like these, we will use the notation $(B)_p \otimes \bigotimes_{i \ne p} A_i = A^{\otimes p-1} \otimes B \otimes A^{\otimes n-p}$.  

\begin{proof}[Proof of \lemref{t:ADVpmwsizecomposition}]
The lemma is a consequence of the composition property \eqnref{e:propertycomposition} from \lemref{t:gamma2_properties}, together with several other properties of the filtered $\gamma_2$ norm.  Let $\digamma = J-F$, so $\ADVpm(f \circ g^n) = \gtwo{\tilde \digamma}{\{ \Delta_{(p,q)} \circ \tilde \digamma \}}$.  We have 
\begin{align*}
\gtwo{\tilde \digamma}{\{ \Delta_{(p,q)} \circ \tilde \digamma \}}
&\leq \gtwo{\tilde \digamma}{\{\tilde \Delta_p \circ \tilde \digamma\}} \max_\rho \gtwo{\tilde \Delta_\rho \circ \tilde \digamma}{\{\Delta_{(p,q)} \circ \tilde \digamma\}} & \text{property \eqnref{e:propertycomposition}} \\
&\leq \gtwo{\digamma}{\{ \Delta_p \circ \digamma \}} \max_\rho \gtwo{\tilde \Delta_\rho}{\{\Delta_{(\rho,q)} : q \in [m]\}} & \text{(\ref{e:propertyduplication}, \ref{e:property6}, \ref{e:propertyunion})} \\
&= \ADVpm(f) \gtwo{J-G}{\Delta} \enspace . & \text{(\ref{e:propertyspecialcases}, \ref{e:propertytensorproduct})}
\end{align*}
The last step uses $\gtwo{(J-G)_\rho \otimes \bigotimes_{i \neq \rho} J_i}{\{ (\Delta_q)_\rho \otimes \bigotimes_{i \neq \rho} J_i\}_q} = \gtwo{J-G}{\Delta} \gtwo{J}{J}^{n-1}$.  
\end{proof}

\begin{proof}[Proof of \lemref{t:compositionlowerbound}]
For the lower bound, we will use the dual formulation of the adversary bound.  Either by writing Eq.~\eqnref{e:ADVpmmax} as an SDP and taking the dual or by noting that $\ADVpm(g) = \max_W \{\langle J-G, W \rangle : \gtwostar{W}{\Delta \circ (J-G)} \le 1\}$ and using Eq.~\eqnref{e:gtwostar}, we find 
\begin{equation} \label{e:neg_adv_dual}
\begin{aligned}
\ADVpm(g)= & \; \underset{\Omega, W}{\text{maximize}}
& & \langle J , W \rangle \\
& \text{subject to}
& & \Omega \circ \identity = \Omega \\
& & & \Tr (\Omega) = 1 \\
& & & W \circ G = 0 \\
& & & \Omega \pm W \circ \Delta_j \succeq 0 
 \enspace .
\end{aligned}
\end{equation}

We first note some basic properties of an optimal dual solution $\Omega, W$.  

\begin{claim} \label{claim:basic_properties}
Let $g: {\mathcal C} \rightarrow D$, where ${\mathcal C} \subseteq C^m$.  Then there is an optimal solution $\Omega, W$ to Eq.~\eqnref{e:neg_adv_dual} that satisfies $\ADVpm(g) \Omega \pm W \succeq 0$.  If $D = \B$ we may also assume $\sum_{x: g(x)=1} \Omega_{x,x} = \sum_{x:g(x)=0} \Omega_{x,x} = \frac{1}{2}$.  
\end{claim}

\begin{proof}
Let $\Omega, W$ be an optimal solution to Eq.~\eqnref{e:neg_adv_dual} and let $d_g = \ADVpm(g) = \langle W, J \rangle$.  Note that $d_g \Omega + W \succeq 0$ if and only if $d_g \Omega - W \succeq 0$ since $\Omega$ is diagonal and $W = W \circ (J-G)$ is bipartite.  Suppose that $d_g \Omega - W \nsucceq 0$.  Then there exists $\phi \succeq 0$, such that $\langle \phi, W \rangle > d_g \langle \phi, \Omega \rangle$.  By normalizing~$\phi$, we may assume that $\langle \phi, \Omega \rangle = 1$.  This shows $\phi \circ \Omega, \phi \circ W$ is a feasible solution for $g$ with objective value greater than $d_g$, a contradiction.  

Now for the second part.   We may reorder the rows and columns of $\Omega, W$ so that all elements $x$ with $g(x) = 0$ come first, then all elements $y$ with $g(y)=1$.  Then the matrices $\Omega \pm W \circ \Delta_i$ have the form 
\begin{equation*}
\begin{bmatrix}
\Omega_0 & 0 \\
0 & \Omega_1
\end{bmatrix}
\pm 
\begin{bmatrix}
0 & X \\
X^\dagger & 0
\end{bmatrix} \circ \Delta_i
 \enspace ,
\end{equation*}
where $W = \big[\begin{smallmatrix}0 & X \\ X^\dagger & 0 \end{smallmatrix}\big]$.  Thus for any $c > 0$, 
\[
\begin{bmatrix}
c\Omega_0 & 0 \\
0 & \tfrac{1}{c}\Omega_1
\end{bmatrix}
\pm 
\begin{bmatrix}
0 & X \\
X^\dagger & 0
\end{bmatrix} \circ \Delta_i
\succeq 0
 \enspace .
\]
If we did not originally have $\Tr(\Omega_0) = \Tr(\Omega_1)$ then choosing $c = \sqrt{\tfrac{\Tr(\Omega_1)}{\Tr(\Omega_0)}}$ to balance them will result in a solution with smaller trace, a contradiction to the optimality of $\Omega, W$.  
\end{proof}
Notice that because of the second item we have that 
\begin{equation}\label{bool_prop}
\sum_{\substack{x,y \\ g(x)=a, g(y)=b}} \ADVpm(g) \Omega_{x,y} +W_{x,y} = \ADVpm(g)/2
\end{equation}
for any $a,b \in \B$.  This is the main property of boolean functions we use.  

Now let $d_f = \ADVpm(f)$, $d_g = \ADVpm(g)$, and let $\Lambda, V$ and $\Omega, W$ be optimal solutions to Eq.~\eqnref{e:neg_adv_dual} for $f$ and $g$, respectively, satisfying the conditions of \claimref{claim:basic_properties} as appropriate.  Our proposed solution to Eq.~\eqnref{e:neg_adv_dual} for the composed function $f \circ g^n$ is the diagonal matrix $d_g^{n-1} \tilde \Lambda \circ \Omega^{\otimes n}$ and weight matrix $\tilde V \circ (d_g \Omega + W)^{\otimes n}$.  Notice that the weight matrix satisfies the constraint $\tilde F \circ \big( \tilde V \circ (d_g \Omega + W)^{\otimes n} \big) = 0$ as $F \circ V = 0$.  

Let us check the objective value.  
\begin{align*}
\big\langle J, (\tilde V \circ (d_g \Omega + W)^{\otimes n} ) \big\rangle
&= \sum_{\substack{a,b \in \B^n \\ f(a) \ne f(b)}} V_{a, b} 
\sum_{\substack{x,y \\ \tilde x=a, \tilde y=b}} \prod_i \big( d_g \Omega_{x^i,y^i} + W_{x^i,y^i} \big) \\
&= \sum_{\substack{a,b \in \B^n \\ f(a) \ne f(b)}} V_{a, b} 
\prod_i \sum_{\substack{x^i,y^i \\ g(x^i)=a_i, g(y^i)=b_i}} \big( d_g \Omega_{x^i,y^i} + W_{x^i,y^i} \big) \\
&= d_f \Big(\frac{d_g}{2}\Big)^n
 \enspace .
\end{align*}
The last line follows by Eq.~\eqnref{bool_prop}.  

It remains to show that $d_g^{n-1} \tilde \Lambda \circ \Omega^{\otimes n} \pm \tilde V \circ (d_g \Omega + W)^{\otimes n} \circ \Delta_{(p,q)} \succeq 0$ for all $(p,q)$.  As $\Tr(d_g^{n-1} \tilde \Lambda \circ \Omega^{\otimes n}) = d_g^{n-1}/2^n$ by \claimref{claim:basic_properties}, this will complete the proof.  

We know that $d_g \Omega + W \succeq 0$, $\Omega + W \circ \Delta_q \succeq 0$.  Also $\tilde \Lambda \pm \tilde V \circ \tilde \Delta_p \succeq 0$ follows from $\Lambda \pm V \circ \Delta_p \succeq 0$ as they are equal up to repetition of some rows and columns.  Thus 
\begin{align*}
0 &\preceq (\tilde \Lambda \pm \tilde V \circ \tilde \Delta_p) \circ \Big((\Omega + W \circ \Delta_q)_p \otimes 
\bigotimes_{i \ne p} (d_g\Omega_i +W_i)\Big) \\
&= d_g^{n-1} \tilde \Lambda \circ \Omega^{\otimes n} 
\pm \tilde V \circ \tilde \Delta_p \circ \Big((\Omega + W \circ \Delta_q)_p \otimes 
\bigotimes_{i \ne p} (d_g \Omega_i +W_i)\Big)
 \enspace .
\end{align*}
This equality follows as $\tilde \Lambda_{x,y} = 0$ unless $\tilde x = \tilde y$, meaning that $g(x^i) = g(y^i)$ for all~$i$.  On the other hand, $W_{x^i,y^i} = 0$ if $g(x^i) = g(y^i)$, which kills all terms involving $\tilde \Lambda$ and $W$.  

Now substitute $\tilde \Delta_p = (J-G)_p \otimes \bigotimes_{i \neq p} J_i$ and simplify $(J-G) \circ (\Omega + W \circ \Delta_q) = (d_g \Omega + W) \circ \Delta_q$ since $(J-G) \circ \Omega = \Delta_q \circ \Omega = G \circ W = 0$.  Since $\Delta_{(p,q)} = (\Delta_q)_p \otimes \bigotimes_{i \neq p} J_i$, this gives $d_g^{n-1} \tilde \Lambda \circ \Omega^{\otimes n} \pm \tilde V \circ (d_g \Omega + W)^{\otimes n} \circ \Delta_{(p,q)} \succeq 0$, as desired.  
\end{proof}

\end{document}